\documentclass[letterpaper,11pt]{article}

\usepackage[runin]{abstract}
\usepackage{geometry}
\usepackage{titlesec}
\usepackage{graphicx}
\usepackage{amsmath}
\usepackage{amsthm}
\usepackage{amsfonts}
\usepackage{amssymb}
\usepackage{empheq}
\usepackage{algorithmic}
\usepackage{algorithm}
\usepackage[authoryear]{natbib}
\usepackage{url}

\usepackage[font=sf]{caption}

\titleformat*{\section}{\Large\bfseries}
\titleformat*{\subsection}{\large\sc}
\titleformat*{\subsubsection}{\itshape}

\usepackage{mathtools}
\DeclarePairedDelimiter\ceil{\lceil}{\rceil}
\DeclarePairedDelimiter\floor{\lfloor}{\rfloor}

\usepackage{epigraph}

\begin{document}

\title{{\bf Evolution toward a Nash equilibrium}}

\author{\large{Ioannis Avramopoulos}}

\date{}

\maketitle

\thispagestyle{empty} 

\newtheorem{definition}{Definition}
\newtheorem{proposition}{Proposition}
\newtheorem{theorem}{Theorem}
\newtheorem*{theorem*}{Theorem}
\newtheorem{corollary}{Corollary}
\newtheorem{lemma}{Lemma}
\newtheorem{axiom}{Axiom}
\newtheorem{thesis}{Thesis}

\vspace*{-0.2truecm}

\begin{abstract}
In this paper, we study the dynamic behavior of Hedge, a well-known algorithm in theoretical machine learning and algorithmic game theory. The empirical average (arithmetic mean) of the iterates Hedge generates is known to converge to a minimax equilibrium in zero-sum games. We generalize that result to show convergence of the empirical average to Nash equilibrium in symmetric bimatrix games (that is bimatrix games where the payoff matrix of each player is the transpose of that of the other) in the sense that every limit point of the sequence of averages is an $\epsilon$-approximate symmetric equilibrium strategy for any desirable $\epsilon$. Our analysis gives rise to a symmetric equilibrium fully polynomial-time approximation scheme, implying {\bf P = PPAD}.
\end{abstract}

\section{Introduction}

\setlength{\epigraphwidth}{0.34\textwidth}
\epigraph{``You want forever, always or never.''}{--- \textup{The Pierces}}

{\em Game theory} is a mathematical discipline concerned with the study of algebraic, analytic, and other objects that abstract the physical world, especially social interactions. The most important {\em solution concept} in game theory is the {\em Nash equilibrium} \citep{Nash}, a strategy profile (combination of strategies) in an $N$-player game such that no unilateral player deviations are profitable. The Nash equilibrium is an attractive solution concept, for example, as Nash showed, an equilibrium is guaranteed to exist in any $N$-person game. Over time this concept has formed a basic cornerstone of {\em economic theory,} but its reach extends beyond economics to the {\em natural sciences} and {\em biology}.

One of the limitations of Nash equilibrium as a plausible solution concept is that we did not have an efficient algorithm for computing one. In fact, it has been conjectured that Nash equilibrium computation is intractable as it is complete for the complexity class {\bf PPAD} \citep{Daskalakis, CDT}. This class, introduced by \cite{PPAD}, contains a variety of related problems (such as computing Brouwer fixed points) that we didn't have efficient algorithms for. Thus, there was a gap between game theory and the theory of computing. In this paper, we take a step toward reconciling these disciplines using {\em dynamical systems theory} as an intermediate step. 

The notion of {\em equilibrium} admits various definitions in the mathematical sciences. One such standard definition is as a {\em fixed point} of a dynamical system. Research in dynamical systems is hardly content with identifying the fixed points of a (continuous) flow or a (discrete) map. What is ultimately important in this mathematical branch is to understand the evolution of system trajectories whether near fixed points (or near, for example, {\em limit cycles}) or globally. The Nash equilibrium can also be understood as a fixed point, for example, of the {\em best response correspondence} of a game (and other dynamics). But dynamical systems (as, for example, studied in theory of learning in games \citep{Fudenberg-Levine, Cesa-Bianchi} or evolutionary game theory \citep{Weibull, PopulationGames}) whose trajectories evolve toward a Nash equilibrium (generically, without restrictive assumptions on the payoff structure) eluded us. Since dynamical systems are algorithms, progress in this direction evidently informs algorithmic research.

There are classes of games where equilibrium computation was already known to be tractable: {\em Zero-sum games} are equivalent to linear programming and thus minimax equilibrium computation admits a polynomial-time algorithm. A bimatrix game is called {\em symmetric} if the payoff matrix of each player is the transpose of that of other. Every symmetric $N$-person game (and, thus every symmetric bimatrix game) admits a symmetric Nash equilibrium \citep{Nash2}.  If the payoff matrix of the symmetric game is also symmetric, then the game is called {\em doubly symmetric}. Symmetric equilibrium computation in doubly symmetric games admits a fully polynomial time approximation scheme (FPTAS) \citep{Ye}. In this paper, we consider the problem of computing a symmetric equilibrium in a symmetric bimatrix game (that is not necessarily doubly symmetric).

\cite{CDT} show that finding a Nash equilibrium in a $2$-person game is a {\bf PPAD}-complete problem and that an equilibrium FPTAS in these games (under either of ``additive" or ``multiplicative" notions of payoff approximation) implies {\bf P = PPAD}. \cite{Avramopoulos2}, drawing on \citep{Jurg}, shows that an FPTAS for a symmetric equilibrium in symmetric bimatrix games also implies that {\bf P=PPAD}. Thus approximating an equilibrium in the class of games we consider is conjectured to be a hard problem. In this paper, we refute that belief. Our result was motivated by an elementary question concerning {\em multiplicative weights dynamics}.

An important result at the intersection of theoretical machine learning and game theory is that the {\em empirical average} of the strategies generated by a well-known multiplicative weights algorithm, namely, Hedge \citep{FreundSchapire1, FreundSchapire2}, approaches the minimax strategy of a respective zero-sum game. In doubly symmetric bimatrix games, Hedge is a multiplicative version of {\em gradient ascent} (with an additional gradient exponentiation step). Gradient ascent is known to converge to critical points of nonlinear optimization problems (for example, see \citep[p. 48]{Bertsekas}) and symmetric equilibrium computation in a doubly symmetric game is a special case of a quadratic programming. It is natural to expect that by tuning Hedge's learning rate parameter, obtaining convergence in doubly symmetric games falls within the realm of possibilities of this algorithm.

But every symmetric bimatrix game is the sum of a doubly symmetric game and a symmetric zero-sum game in the sense that every (payoff) matrix $C$ can be decomposed as follows: 
\begin{align*}
C = \frac{1}{2} (C+C^T) + \frac{1}{2} (C-C^T).
\end{align*}
Since the iterates of Hedge are expected to converge to a symmetric equilibrium in a doubly symmetric game (which would imply that the empirical average of the iterates would converge likewise) and the empirical average of the iterates is known to converge to such an equilibrium in symmetric zero-sum games,
is it then possible that the empirical average of iterates converges to a symmetric equilibrium in the general class of symmetric bimatrix games? Our answer is affirmative.

\subsection{Our results and techniques}

We show in Theorem \ref{asymptotic_convergence_theorem_3} that under a fixed learning rate parameter $\alpha$ every limit point of the sequence of empirical averages is an $\epsilon$-approximate symmetric equilibrium strategy for any desirable $\epsilon$. We achieve that by tuning parameter $\alpha$. A corollary is that a symmetric equilibrium always exists in symmetric bimatrix games. We may thus obtain alternative (constructive) proofs of existence of equilibria and fixed points in a variety of related problems (such as $N$-person games). Our analysis gives an equilibrium fully polynomial time approximation scheme, which implies {\bf P=PPAD}. 

The basis of our analysis is a formula for the equilibrium approximation error (cf. Lemma \ref{basis_lemma_3}), which we derive directly from the equation defining Hedge. We pierce through this formula in Lemma \ref{LA} using an inductive proof that makes use of the ``multiplicative weights convexity lemma.'' (Lemma \ref{LA} was the hardest part of the analysis required to obtain equilibrium approximation bounds.) The multiplicative weights convexity lemma states that the composition of the relative entropy function with Hedge is a convex function of $\alpha$ a result appears for the first time in \citep{Avramopoulos1}.\footnote{In this manuscript, I erroneously believed to have shown {\bf P = PPAD}. The error is in Lemma 10.} Using this lemma we obtain in Lemma \ref{LA} a key recursive relationship using logarithms that we make use of in our inductive proof. Our use of the logarithmic function is related to a proof that the time average of the trajectory of the replicator dynamic in a symmetric bimatrix game with an interior equilibrium converges to that equilibrium \citep{Wolff} (see also \citep[p. 91]{Weibull}). 
 That the relative entropy function (also known as Kullback-Leibler divergence) facilitates analyzing multiplicative weights (and its continuous approximation, namely, the {\em replicator dynamic}) is well known (for example, see \citep{Weibull, FreundSchapire2}).

\subsection{Other related work}

There is a significant amount of work on Nash equilibrium computation especially in the setting of $2$-player games. We simply mention a boundary of those results. The Lemke-Howson algorithm for computing an equilibrium in a bimatrix game is considered by many to be the state of the art in exact equilibrium computation but it has been shown to run in exponential time in the worst case \citep{Savani-Stengel}. There is a quasi-polynomial algorithm for additively approximate Nash equilibria in bimatrix games due to \cite{LMM} (based on looking for equilibria of small support on a grid). Prior to our work, the best polynomial-time approximation algorithm for a Nash equilibrium achieves a 0.3393 approximation \citep{Tsaknakis-Spirakis-journal}.

\subsection{Outline of the rest of the paper}

Section \ref{preliminaries} starts off with preliminaries in equilibrium computation discussing bimatrix and symmetric bimatrix games, defining Hedge, and proving a key relationship used in our main results as an implication of the multiplicative weights convexity lemma. Section \ref{fixed_rate_convergence} contains our main results on asymptotic convergence to an equilibrium as well as on equilibrium approximation. Our main result that {\bf P = PPAD} is obtained via an equilibrium fully polynomial-time approximation scheme.

\section{Equilibrium computation background and preliminary results}
\label{preliminaries}

\subsection{Symmetric bimatrix games}

A $2$-player (bimatrix) game in normal form is specified by a pair of $n \times m$ matrices $A$ and $B$, the former corresponding to the {\em row player} and the latter to the {\em column player}. A {\em mixed strategy} for the row player is a probability vector $P \in \mathbb{R}^n$ and a mixed strategy for the column player is a probability vector $Q \in \mathbb{R}^m$. The {\em payoff} to the row player of $P$ against $Q$ is $P \cdot A Q$ and that to the column player is $P \cdot B Q$. Let us denote the space of probability vectors for the row player by $\mathbb{P}$ and the corresponding space for the column player by $\mathbb{Q}$. A Nash equilibrium of a $2$-player game $(A, B)$ is a pair of mixed strategies $P^*$ and $Q^*$ such that all unilateral deviations from these strategies are not profitable, that is, for all $P \in \mathbb{P}$ and $Q \in \mathbb{Q}$, we simultaneously have that
\begin{align}
P^* \cdot AQ^* &\geq P \cdot AQ^*\label{eqone}\\
P^* \cdot BQ^* &\geq P^* \cdot BQ.\label{eqtwo}
\end{align}
If $B = A^T$, where $A^T$ is the transpose, $(A, B)$ is called {\em symmetric}. Let $(C, C^T)$ be a symmetric bimatrix game. We call an equilibrium $(P^*, Q^*)$  symmetric if $P^* = Q^*$. If $(X^*, X^*)$ is a symmetric equilibrium, we call $X^*$ a symmetric equilibrium strategy. $C$ denotes the symmetric game $(C, C^T)$.

\subsubsection{Some further notation}

We denote the space of symmetric bimatrix games by $\mathbb{C}$. If the payoff entries lie in the range $[0, 1]$, we denote the corresponding space by $\mathbb{\hat{C}}$. Given $C \in \mathbb{C}$, we denote the corresponding set of pure strategies by $\mathcal{K}(C) = \{1, \ldots, n\}$. Pure strategies are denoted either as $i$ or as $E_i$, a probability vector whose mass is concentrated in position $i$. $\mathbb{X}(C)$ is the probability simplex (space of mixed strategies) corresponding to $C \in \mathbb{C}$. We denote the (relative) interior of $\mathbb{X}(C)$ by $\mathbb{\mathring{X}}(C)$ (every pure strategy in $\mathbb{\mathring{X}}(C)$ has probability mass). Let $X \in \mathbb{X}(C)$. We define the {\em support} or {\em carrier} of $X$ by
\begin{align*}
\mathcal{C}(X) \equiv \{ i \in \mathcal{K}(C) | X(i) > 0\}.
\end{align*}

\subsubsection{Approximate equilibria}

Conditions \eqref{eqone} and \eqref{eqtwo} simplify as follows for a symmetric equilibrium strategy $X^*$:
\begin{align*}
\forall X \in \mathbb{X}(C) : (X^* - X) \cdot CX^* \geq 0.
\end{align*}
An $\epsilon$-approximate symmetric equilibrium satisfies:
\begin{align*}
\forall X \in \mathbb{X}(C) : (X^* - X) \cdot CX^* \geq -\epsilon.
\end{align*}
We may equivalently write the previous expression as
\begin{align*}
(CX^*)_{\max} - X^* \cdot CX^* \leq \epsilon,
\end{align*}
where
\begin{align*}
(CX^*)_{\max} = \max\{ Y \cdot CX^* | Y \in \mathbb{X}(C) \}.
\end{align*}
If $Y^+ \cdot CX^* = (CX^*)_{\max}$, then $Y^+$ is called a best response to $X^*$.

\subsection{The convexity lemma of multiplicative weights and implications}
\label{convexity_lemma_section}

Hedge \citep{FreundSchapire1, FreundSchapire2} induces the following map in our setting:
\begin{align}
T_i(X) = X(i) \cdot \frac{\exp\left\{ \alpha E_i \cdot CX \right\}}{ \sum_{j=1}^n X(j) \exp \left\{ \alpha E_j \cdot CX \right\} }, \quad i = 1, \ldots, n,\label{main_exp}
\end{align}
where $C$ is the payoff matrix of a symmetric bimatrix game, $n$ is the number of pure strategies, $E_i$ is the probability vector corresponding to pure strategy $i$, and $X(i)$ is the probability mass of pure strategy $i$. Parameter $\alpha$ is called the {\em learning rate}, which has the role of a {\em step size} in our equilibrium computation setting. In this paper, we do not study the behavior of the iterates that $T$ generates per se but instead the sequence $\left\{ \bar{X}^K \right\}_{K = 0}^{\infty}$ of empirical averages of the iterates that $T$ generates starting from an interior to the probability simplex strategy. The empirical average $\bar{X}^K$ at iteration $K = 0, 1, 2, \ldots$ is a simple arithmetic mean
\begin{align*}
\bar{X}^K = \frac{1}{K+1} \sum_{k=0}^K X^k.
\end{align*}

\subsubsection{The multiplicative weights convexity lemma}

Let us now give some preliminary results on Hedge dynamics. Part of our analysis of Hedge relies on the relative entropy function between probability distributions (also called {\em Kullback-Leibler divergence}). The relative entropy between the $n \times 1$ probability vectors $P > 0$ (that is, for all $i = 1, \ldots, n$, $P(i) > 0$) and $Q > 0$ is given by 
\begin{align*}
RE(P, Q) \doteq \sum_{i=1}^n P(i) \ln \frac{P(i)}{Q(i)}.
\end{align*}
However, this definition can be relaxed: The relative entropy between $n \times 1$ probability vectors $P$ and $Q$ such that, given $P$, for all $Q \in \{ \mathcal{Q} \in \mathbb{X} | \mathcal{C}(P) \subset \mathcal{C}(\mathcal{Q}) \}$, is
\begin{align*}
RE(P, Q) \doteq \sum_{i \in \mathcal{C}(P)} P(i) \ln \frac{P(i)}{Q(i)}.
\end{align*}
We refer to  \cite[p.96]{Weibull} from well-known properties of the relative entropy function. With this background in mind, we state the following lemma (which we refer to as the multiplicative weights convexity lemma) generalizing \cite[Lemma 2]{FreundSchapire2}.

\begin{lemma}[\citep{Avramopoulos2}]
\label{convexity_lemma}
Let $T$ be as in \eqref{main_exp}. Then
\begin{align*}
\forall X \in \mathbb{\mathring{X}}(C) \mbox{ } \forall Y \in \mathbb{X}(C) : RE(Y, T(X)) \mbox{ is a convex function of }\alpha.
\end{align*}
\end{lemma}

The next lemma is shown in \citep{Avramopoulos2}. We repeat the proof for completeness. In the proof, we use the following ``secant inequality'' for a convex function $F(\cdot)$ and its derivative $F'(\cdot)$:
\begin{align}
\forall \mbox{ } b > a : F'(a) \leq \frac{F(b) - F(a)}{b - a} \leq F'(b).\label{secant_inequality}
\end{align}

\begin{lemma}
\label{convexity_lemma_normalized}
Let $C \in \mathbb{\hat{C}}$. Then, for all $Y \in \mathbb{X}(C)$ and for all $X \in \mathbb{\mathring{X}}(C)$, we have that
\begin{align*}
\forall \alpha > 0 : RE(Y, T(X)) \leq RE(Y, X) - \alpha (Y-X) \cdot CX + \alpha (\exp\{\alpha\} - 1).
\end{align*}
\end{lemma}

\begin{proof}
Since, by Lemma \ref{convexity_lemma}, $RE(Y, T(X)) - RE(Y, X)$ is a convex function of $\alpha$, we have by the aforementioned secant inequality that, for $\alpha > 0$,
\begin{align}
RE(Y, T(X)) - RE(Y, X) \leq \alpha \left( RE(Y, T(X)) - RE(Y, X) \right)' = \alpha \cdot \frac{d}{d \alpha} RE(Y, T(X))\label{ooone}
\end{align}
where it can be readily computed that
\begin{align*}
\frac{d}{d \alpha} RE(Y, T(X)) = \frac{\sum_{j = 1}^n X(j) (CX)_j \exp\{ \alpha (CX)_j \}}{\sum_{j = 1}^n X(j) \exp\{ \alpha (CX)_j \}} - Y \cdot CX.
\end{align*}
Using Jensen's inequality in the previous expression, we obtain 
\begin{align}
\frac{d}{d \alpha} RE(Y, T(X)) \leq \frac{\sum_{j = 1}^n X(j) (CX)_j \exp\{ \alpha (CX)_j \}}{\exp\{ \alpha X \cdot CX \}} - Y \cdot CX.\label{vbvbvb}
\end{align}
Note now that
\begin{align*}
\exp\{ \alpha x \} \leq 1 + (\exp\{ \alpha \} - 1) x, x \in [0, 1],
\end{align*}
an inequality used in \citep[Lemma 2]{FreundSchapire2}. Using the latter inequality, we obtain from \eqref{vbvbvb} that
\begin{align*}
\frac{d}{d \alpha} RE(Y, T(X)) \leq \frac{X \cdot CX}{\exp\{ \alpha X \cdot CX \}} - Y \cdot CX + (\exp\{ \alpha \} - 1) \frac{\sum_{j=1}^n X(j) (CX)_j^2}{\exp\{ \alpha X \cdot CX \}}
\end{align*}
and since $\exp\{\alpha X \cdot CX\} \geq 1$ again by the assumption that $C \in \mathbb{\hat{C}}$, we have
\begin{align*}
\frac{d}{d \alpha} RE(Y, T(X)) \leq X \cdot CX - Y \cdot CX + (\exp\{\alpha \} - 1)  \sum_{j=1}^n X(j) (CX)_j^2.
\end{align*}
Noting that $\sum X(j) (CX)_j^2 \leq 1$ and combining with \eqref{ooone} yields the lemma.
\end{proof}

As a corollary to the previous lemma, we obtain the following lemma:
\begin{lemma}
\label{very_useful_lemma}
Let $C \in \mathbb{\hat{C}}$. Then, for all $i \in \{1, \ldots, n\}$ and for all $X \in \mathbb{\mathring{X}}(C)$, we have that
\begin{align*}
\forall \alpha > 0 : \ln(T_i(X)) \geq \ln(X(i)) + \alpha (E_i-X) \cdot CX - \alpha (\exp\{\alpha\} - 1).
\end{align*}
\end{lemma}

\begin{proof}
This lemma is a simple implication of Lemma \ref{convexity_lemma_normalized} noting that $RE(E_i, X) = - \ln(X(i))$.
\end{proof}

We note that \eqref{ooone} can be obtained from Slater's inequality as follows: Let $\hat{X} \equiv T(X)$. Then
\begin{align*}
\frac{\hat{X}(i)}{X(i)} = \frac{\exp\{ \alpha (CX)_i \}}{\sum_{j=1}^n X(j) \exp\{ \alpha (CX)_j \}}.
\end{align*}
Slater's inequality (cf. \citep{Dragomir}) gives
\begin{align*}
\sum_{j=1}^n X (j) \exp\{\alpha (CX)_j\} \leq \exp\left\{ \alpha \frac{\sum_{j=1}^n X(j) (CX)_j \exp\{\alpha (CX)_j\}}{\sum_{j=1}^n X(j) \exp\{\alpha (CX)_j\}} \right\}.
\end{align*}
Combining the previous inequalities, we obtain
\begin{align*}
\frac{\hat{X}(i)}{X(i)} \geq \exp\left\{ \alpha \left( (CX)_i - \frac{\sum_{j=1}^n X(j) (CX)_j \exp\{\alpha (CX)_j\}}{\sum_{j=1}^n X(j) \exp\{\alpha (CX)_j\}} \right) \right\}.
\end{align*}
Taking logarithms
\begin{align*}
\ln(\hat{X}(i)) - \ln(X(i)) \geq \alpha \left( (CX)_i - \frac{\sum_{j=1}^n X(j) (CX)_j \exp\{\alpha (CX)_j\}}{\sum_{j=1}^n X(j) \exp\{\alpha (CX)_j\}} \right).
\end{align*}
Multiplying both sides with $Y(i)$ and summing over $i=1, \ldots, n$
\begin{align*}
RE(Y, X) - RE(Y, \hat{X}) \geq \alpha \left( Y \cdot CX - \frac{\sum_{j=1}^n X(j) (CX)_j \exp\{\alpha (CX)_j\}}{\sum_{j=1}^n X(j) \exp\{\alpha (CX)_j\}} \right).
\end{align*}
Rearranging
\begin{align*}
RE(Y, \hat{X}) - RE(Y, X) \leq \alpha \left( \frac{\sum_{j=1}^n X(j) (CX)_j \exp\{\alpha (CX)_j\}}{\sum_{j=1}^n X(j) \exp\{\alpha (CX)_j\}} - Y \cdot CX \right),
\end{align*}
and, thus,
\begin{align*}
RE(Y, \hat{X}) - RE(Y, X) \leq \alpha \frac{d}{d \alpha} RE(Y, T(X)).
\end{align*}

\section{Convergence under a fixed learning rate and P=PPAD}
\label{fixed_rate_convergence}

\subsection{Our asymptotic convergence result}

\begin{theorem}
\label{asymptotic_convergence_theorem_3}
Let $C \in \mathbb{\hat{C}}$ and $X^k \equiv T^k(X^0)$, where $X^0 \in \mathbb{\mathring{X}}(C)$ is the uniform distribution. Assume the learning rate $\alpha > 0$ is constant. Then considering the sequence of empirical averages
\begin{align*}
\left\{ \frac{1}{K+1} \sum_{k=0}^K X^k \equiv \bar{X}^K \right\}_{K=0}^{\infty},
\end{align*}
every limit point of this sequence is an $n(\exp\{\alpha\}-1)$-approximate symmetric equilibrium strategy of $C$, where $n$ is the number of pure strategies of $C$.
\end{theorem}

\noindent
The assumption $X^0$ is the uniform distribution is to simplify notation. It is a simple exercise to lift that assumption. Note that the definition of the empirical average implies the recursive relationship
\begin{align*}
\bar{X}^{K+1} = \frac{1}{K+2} X^{K+1} + \frac{K+1}{K+2} \bar{X}^K,
\end{align*}
which we make use of in the sequel.

\begin{lemma}
\label{basis_lemma}
Suppose $X^0$ is the uniform distribution. Then
\begin{align*}
(C\bar{X}^K)_{\max}  - \bar{X}^K \cdot C\bar{X}^K \leq - \frac{1}{\alpha (K+1)} \sum_{j=1}^n \bar{X}^K(j) \ln \left( X^{K+1}(j) \right)
\end{align*}
\end{lemma}

\begin{proof}
Let $T(X) \equiv \hat{X}$. Then straight algebra gives
\begin{align*}
\frac{\hat{X}(i)}{\hat{X}(j)} = \frac{X(i)}{X(j)} \exp\{\alpha ((CX)_i - (CX)_j)\}
\end{align*}
and taking logarithms on both sides we obtain
\begin{align*}
\ln \left( \frac{\hat{X}(i)}{\hat{X}(j)} \right) = \ln \left( \frac{X(i)}{X(j)} \right) + \alpha ((CX)_i - (CX)_j).
\end{align*}
We may write the previous equation as
\begin{align*}
\ln \left( \frac{X^{k+1}(i)}{X^{k+1}(j)} \right) = \ln \left( \frac{X^k(i)}{X^k(j)} \right) + \alpha ((CX^k)_i - (CX^k)_j)
\end{align*}
Summing over $k = 0, \ldots K$, we obtain
\begin{align*}
\ln \left( \frac{X^{K+1}(i)}{X^{K+1}(j)} \right) = \ln \left( \frac{X^0(i)}{X^0(j)} \right) + \alpha \sum_{k=0}^K ((CX^k)_i - (CX^k)_j)
\end{align*}
and dividing by $K+1$ and rearranging, we further obtain
\begin{align*}
\frac{1}{\alpha (K+1)} \ln \left( \frac{X^{K+1}(i)}{X^{K+1}(j)} \right) = \frac{1}{\alpha (K+1)} \ln \left( \frac{X^0(i)}{X^0(j)} \right) + (E_i - E_j) \cdot C\bar{X}^K.
\end{align*}
Under the assumption $X^0$ is the uniform distribution,
\begin{align*}
\frac{1}{\alpha (K+1)} \ln \left( \frac{X^{K+1}(i)}{X^{K+1}(j)} \right) = (E_i - E_j) \cdot C\bar{X}^K,
\end{align*}
which implies
\begin{align*}
\frac{1}{\alpha (K+1)} \ln \left( \frac{X^{K+1}(i_{\max})}{X^{K+1}(j)} \right) = (C\bar{X}^K)_{\max} - (C\bar{X}^K)_j
\end{align*}
and further implies
\begin{align*}
(C\bar{X}^K)_{\max} - (C\bar{X}^K)_j \leq - \frac{1}{\alpha (K+1)} \ln \left( X^{K+1}(j) \right).
\end{align*}
Multiplying both sides with $\bar{X}^K(j)$ and summing over $j=1, \ldots, n$, we finally obtain
\begin{align*}
(C\bar{X}^K)_{\max} - \bar{X}^K \cdot C\bar{X}^K \leq - \frac{1}{\alpha (K+1)} \sum_{j=1}^n \bar{X}^K(j) \ln \left( X^{K+1}(j) \right)
\end{align*}
as claimed.
\end{proof}

\begin{lemma}
\label{basis_lemma_2}
Suppose $X^0$ is the uniform distribution. Then
\begin{align*}
(C\bar{X}^K)_{\max}  - \bar{X}^{K+1} \cdot C\bar{X}^K \leq - \frac{1}{\alpha (K+1)} \sum_{j=1}^n \bar{X}^{K+1}(j) \ln \left( X^{K+1}(j) \right)
\end{align*}
\end{lemma}

\begin{proof}
The proof is directly analogous to the proof of Lemma \ref{basis_lemma}.
\end{proof}

\begin{lemma}
\label{basis_lemma_3}
Under the assumptions of Theorem \ref{asymptotic_convergence_theorem_3},
\begin{align*}
(C\bar{X}^K)_{\max}  - \bar{X}^{K} \cdot C\bar{X}^K \leq - \frac{1}{\alpha (K+1)} \sum_{j=1}^n \bar{X}^{K+1}(j) \ln \left( X^{K+1}(j) \right) + \frac{1}{K+1}.
\end{align*}
\end{lemma}

\begin{proof}
Since
\begin{align*}
\bar{X}^{K+1} = \frac{1}{K+2} X^{K+1} + \frac{K+1}{K+2} \bar{X}^K
\end{align*}
we have
\begin{align*}
\bar{X}^{K+1} \cdot C \bar{X}^{K} = \frac{1}{K+2} X^{K+1} \cdot C\bar{X}^K + \frac{K+1}{K+2} \bar{X}^K \cdot C\bar{X}^K
\end{align*}
and, therefore,
\begin{align*}
\bar{X}^K \cdot C\bar{X}^K = \frac{K+2}{K+1}  \bar{X}^{K+1} \cdot C \bar{X}^{K} - \frac{1}{K+1} X^{K+1} \cdot C\bar{X}^K,
\end{align*}
which, using the assumption $C \in \mathbb{\hat{C}}$ so that $X^{K+1} \cdot C\bar{X}^K \leq 1$, implies
\begin{align*}
\bar{X}^K \cdot C\bar{X}^K \geq \frac{K+2}{K+1}  \bar{X}^{K+1} \cdot C \bar{X}^{K} - \frac{1}{K+1} \geq \bar{X}^{K+1} \cdot C \bar{X}^{K} - \frac{1}{K+1},
\end{align*}
which, combined with Lemma \ref{basis_lemma_2}, implies the lemma.
\end{proof}

\newpage

\begin{lemma}
\label{LA}
Under the assumptions of Theorem \ref{asymptotic_convergence_theorem_3}, the sequence $\{ X^k(j) \}$ of probability masses corresponding to pure strategy $j \in \{1, \ldots, n\}$ satisfies the following relation
\begin{align*}
- \frac{1}{\alpha (K+1)} \bar{X}^K(j) \ln \left( X^{K}(j) \right) \leq &- \frac{1}{\alpha (K+1)} \left( \frac{1}{K+1} \sum_{k=0}^K X^k(j) \ln \left( X^k(j) \right) \right) + (\exp\{\alpha\} - 1) + \rho^K,
\end{align*}
where $\rho = 1/2$.
\end{lemma}

\begin{proof}
Our proof is by induction. The basis of the induction ($K = 0$) is straightforward. Suppose now the relation is true for iteration $K > 0$. We then have for the next iteration
\begin{align*}
- \frac{1}{\alpha (K+2)} \bar{X}^{K+1}(j) \ln \left( X^{K+1}(j) \right) =
\end{align*}
\begin{align*}
= - \frac{1}{\alpha (K+2)} \left( \frac{1}{K+2} X^{K+1}(j) + \frac{K+1}{K+2} \bar{X}^K(j) \right) \ln \left( X^{K+1}(j) \right) =
\end{align*}
Using straight algebra we obtain
\begin{align*}
= - \frac{1}{\alpha (K+2)}  \frac{1}{K+2} X^{K+1}(j) \ln \left( X^{K+1}(j) \right) - \frac{1}{\alpha (K+2)} \frac{K+1}{K+2} \bar{X}^K(j) \ln \left( X^{K+1}(j) \right) =
\end{align*}
\begin{align}
= - \frac{1}{\alpha (K+2)}  \frac{1}{K+2} X^{K+1}(j) \ln \left( X^{K+1}(j) \right) - \frac{(K+1)^2}{(K+2)^2} \cdot \frac{1}{\alpha (K+1)}  \bar{X}^K(j) \ln \left( X^{K+1}(j) \right)\label{abcde}
\end{align}
We may now make use of the relation
\begin{align*}
\ln(X^{k+1}(j)) \geq \ln(X^k(j)) + \alpha ((CX^k)_{j} - X^k \cdot CX^k) - \alpha (\exp\{\alpha\}-1)
\end{align*}
which is obtained from Lemma \ref{very_useful_lemma}. Since $(CX^k)_{j} - X^k \cdot CX^k \geq -1$ by the assumption $C \in \mathbb{\hat{C}}$, using straight algebra we obtain
\begin{align*}
\ln(X^{k+1}(j)) \geq \ln(X^k(j)) - \alpha \exp\{\alpha\}
\end{align*}
Combining the previous inequality with \eqref{abcde}, we obtain
\begin{align*}
\leq &- \frac{1}{\alpha (K+2)}  \frac{1}{K+2} X^{K+1}(j) \ln \left( X^{K+1}(j) \right) - \frac{(K+1)^2}{(K+2)^2} \cdot \frac{1}{\alpha (K+1)}  \bar{X}^K(j) \ln \left( X^{K}(j) \right) +\\
&+ \frac{(K+1)^2}{(K+2)^2} \frac{1}{\alpha(K+1)} \alpha \exp\{\alpha\} \bar{X}^K(j)
\end{align*}
Using the induction hypothesis we further obtain
\begin{align*}
\leq &- \frac{1}{\alpha (K+2)}  \frac{1}{K+2} X^{K+1}(j) \ln \left( X^{K+1}(j) \right) - \frac{(K+1)^2}{(K+2)^2}\frac{1}{\alpha (K+1)} \left( \frac{1}{K+1} \sum_{k=0}^K X^k(j) \ln \left( X^k(j) \right) \right) +\\
  &+ \frac{(K+1)^2}{(K+2)^2} \frac{1}{\alpha(K+1)} \alpha \exp\{\alpha\} \bar{X}^K(j) + \frac{(K+1)^2}{(K+2)^2} (\exp\{\alpha\} - 1 + \rho^{K})
\end{align*}
which implies (since $0 < \bar{X}^K(j) < 1$)
\begin{align*}
\leq &- \frac{1}{\alpha (K+2)}  \frac{1}{K+2} X^{K+1}(j) \ln \left( X^{K+1}(j) \right) - \frac{(K+1)^2}{(K+2)^2}\frac{1}{\alpha (K+1)} \left( \frac{1}{K+1} \sum_{k=0}^K X^k(j) \ln \left( X^k(j) \right) \right) +\\
  &+ \frac{(K+1)^2}{(K+2)^2} \frac{1}{K+1} \exp\{\alpha\} + \frac{(K+1)^2}{(K+2)^2} (\exp\{\alpha\} -1 + \rho^{K})
\end{align*}
which further implies
\begin{align*}
\leq &- \frac{1}{\alpha (K+2)}  \frac{1}{K+2} X^{K+1}(j) \ln \left( X^{K+1}(j) \right) - \frac{1}{\alpha (K+2)} \left( \frac{1}{K+2} \sum_{k=0}^K X^k(j) \ln \left( X^k(j) \right) \right) +\\
  &+ \frac{K+1}{K+2} \exp\{\alpha\} - \left( \frac{K+1}{K+2} \right)^2 (1- \rho^K)
\end{align*}
and since
\begin{align*}
\frac{K+1}{K+2} \geq 1- \rho, K = 0,1, \ldots
\end{align*}
we obtain
\begin{align*}
\leq &- \frac{1}{\alpha (K+2)}  \frac{1}{K+2} X^{K+1}(j) \ln \left( X^{K+1}(j) \right) - \frac{1}{\alpha (K+2)} \left( \frac{1}{K+2} \sum_{k=0}^K X^k(j) \ln \left( X^k(j) \right) \right) +\\
  &+ \frac{K+1}{K+2} (\exp\{\alpha\}- 1 - \rho - \rho^K + \rho^{K+1})
\end{align*}
which finally implies
\begin{align*}
\leq -\frac{1}{\alpha (K+2)} \left( \frac{1}{K+2} \sum_{k=0}^{K+1} X^k(j) \ln \left( X^k(j) \right) \right) + ( \exp\{\alpha\}- 1) + \rho^{K+1}
\end{align*}
completing the proof of the lemma.
\end{proof}

\begin{lemma}
\label{cat}
Under the assumptions of Theorem \ref{asymptotic_convergence_theorem_3}, the sequence of empirical averages
\begin{align*}
\left\{ \frac{1}{K+1} \sum_{k=0}^K X^k \equiv \bar{X}^K \right\}_{K=0}^{\infty},
\end{align*}
satisfies the following relation
\begin{align}
(C\bar{X}^K)_{\max}  - \bar{X}^K \cdot C\bar{X}^K \leq \frac{n \exp\{-1\}}{\alpha (K+1)} + \frac{1}{K+1} + \left( \frac{K+2}{K+1} \right) n (\exp\{\alpha\}-1) + \left( \frac{K+2}{K+1} \right) n \rho^K.\label{mylove}
\end{align}
\end{lemma}

\begin{proof}
Lemma \ref{basis_lemma_3} implies that
\begin{align}
(C\bar{X}^K)_{\max}  - \bar{X}^K \cdot C\bar{X}^K \leq - \frac{1}{\alpha (K+1)} \sum_{j=1}^n \bar{X}^{K+1}(j) \ln \left( X^{K+1}(j) \right) + \frac{1}{K+1}.\label{babe}
\end{align}
Lemma \ref{LA} further implies that
\begin{align*}
- \frac{1}{\alpha (K+2)} \bar{X}^{K+1}(j) \ln \left( X^{K+1}(j) \right) \leq - \frac{1}{\alpha (K+2)} \left( \frac{1}{K+2} \sum_{k=0}^{K+1} X^k(j) \ln \left( X^k(j) \right) \right) + (\exp\{\alpha\}-1) + \rho^K.
\end{align*}
Using elementary calculus we obtain $- x \ln(x) \leq \exp\{-1\}, x \in [0, 1]$ and, therefore,
\begin{align*}
- \frac{1}{\alpha (K+2)} \bar{X}^{K+1}(j) \ln \left( X^{K+1}(j) \right) \leq \frac{\exp\{-1\}}{\alpha (K+2)} + (\exp\{\alpha\}-1) + \rho^K.
\end{align*}
Thus,
\begin{align*}
- \frac{1}{\alpha (K+1)} \sum_{j=1}^n \bar{X}^{K+1}(j) \ln \left( X^{K+1}(j) \right) \leq \frac{n \exp\{-1\}}{\alpha (K+1)} + \left( \frac{K+2}{K+1} \right) n (\exp\{\alpha\}-1) + n \left( \frac{K+2}{K+1} \right) \rho^K
\end{align*}
and combining with \eqref{babe} we obtain
\begin{align*}
(C\bar{X}^K)_{\max}  - \bar{X}^K \cdot C\bar{X}^K \leq \frac{n \exp\{-1\}}{\alpha (K+1)} + \frac{1}{K+1} + \left( \frac{K+2}{K+1} \right) n (\exp\{\alpha\}-1) + n \left( \frac{K+2}{K+1} \right) \rho^K
\end{align*}
as claimed.
\end{proof}

\begin{proof}[Proof of Theorem \ref{asymptotic_convergence_theorem_3}]
\eqref{mylove} implies that
\begin{align*}
\lim_{K \rightarrow \infty} \left\{ (C\bar{X}^K)_{\max}  - \bar{X}^K \cdot C\bar{X}^K \right\} = n (\exp\{\alpha\}-1).
\end{align*}
Therefore, if $\bar{X}$ is any limit point of $\{\bar{X}^K\}$, then 
\begin{align*}
(C\bar{X})_{\max} - \bar{X} \cdot C\bar{X} \leq n (\exp\{\alpha\}-1),
\end{align*}
implying it is an $n (\exp\{\alpha\}-1)$-approximate equilibrium strategy and the theorem follows.
\end{proof}

\subsection{A fully polynomial-time approximation scheme}

Deriving an equilibrium fully polynomial-time approximation scheme from relation \eqref{mylove} is straightforward as shown next. First we need a lemma:

\begin{lemma}
Under the assumptions of Theorem \ref{asymptotic_convergence_theorem_3}, for all $\theta > 0$, in 
\begin{align}
K = \floor*{\frac{n \exp\{-1\}+1 +\alpha }{\alpha \theta}}\label{K_formula}
\end{align}
iterations, we have that
\begin{align*}
(C\bar{X}^K)_{\max}  - \bar{X}^K \cdot C\bar{X}^K \leq \left( \frac{K+2}{K+1} \right) n(\exp\{\alpha\}-1) + \left( \frac{K+2}{K+1} \right) n \rho^K + \theta.
\end{align*}
\end{lemma}

\begin{proof}
Let us repeat \eqref{mylove} for convenience:
\begin{align*}
(C\bar{X}^K)_{\max}  - \bar{X}^K \cdot C\bar{X}^K \leq \frac{n \exp\{-1\}}{\alpha (K+1)} + \frac{1}{K+1} + \left( \frac{K+2}{K+1} \right) n (\exp\{\alpha\}-1) + \left( \frac{K+2}{K+1} \right) n \rho^K.
\end{align*}
Letting
\begin{align*}
\theta = \frac{n \exp\{-1\}}{\alpha (K+1)} + \frac{1}{K+1}
\end{align*}
implies that $K$ should be the first integer such that
\begin{align*}
K+1 \geq \frac{n \exp\{-1\}}{\alpha \theta} + \frac{1}{\theta} = \frac{n \exp\{-1\}+1 +\alpha }{\alpha \theta}.
\end{align*}
Since
\begin{align*}
\floor*{\frac{n \exp\{-1\}+1 +\alpha }{\alpha \theta}} + 1 \geq \frac{n \exp\{-1\}+1 +\alpha }{\alpha \theta}
\end{align*}
the lemma follows.
\end{proof}

\begin{theorem}
Under the assumptions of Theorem \ref{asymptotic_convergence_theorem_3}, to attain an equilibrium approximation error that is at most $\epsilon$, we need at most
\begin{align*}
K = \floor*{\frac{n \exp\{-1\}+1 +\ln\left(1+\epsilon'/(3n)\right) }{(\epsilon'/3) \ln\left(1+\epsilon'/(3n)\right)}}
\end{align*}
iterations, for some $\epsilon' < \epsilon$, which can be readily computed a priori.
\end{theorem}

\begin{proof}
The number of iterations in the statement of the theorem is obtained by letting
\begin{align*}
\alpha = \ln\left(1+\frac{\epsilon'}{3n}\right) \qquad \theta = \frac{\epsilon'}{3}
\end{align*}
in \eqref{K_formula}. Since the first integer $K'$ such that
\begin{align*}
n \rho^{K'} \leq \frac{\epsilon'}{3}
\end{align*}
is
\begin{align*}
K' = \ceil*{\ln \left( \frac{\epsilon'}{3n} \right) / \ln(\rho)}
\end{align*}
which is smaller than $K$, these values imply an equilibrium approximation error of at most
\begin{align*}
\left( \frac{K+2}{K+1} \right) \frac{\epsilon'}{3} + \frac{\epsilon'}{3} + \left( \frac{K+2}{K+1} \right) \frac{\epsilon'}{3}.
\end{align*}
We are looking for $\epsilon'$ such that
\begin{align*}
\left( \frac{K+2}{K+1} \right) \frac{\epsilon'}{3} + \frac{\epsilon'}{3} + \left( \frac{K+2}{K+1} \right) \frac{\epsilon'}{3} \leq \epsilon.
\end{align*}
To that end, we choose a lower bound $\hat{K}$ on the number of iterations, for example,
\begin{align*}
\hat{K} = \floor*{\frac{n \exp\{-1\}+1 +\ln\left(1+\epsilon/(3n)\right) }{(\epsilon/3) \ln\left(1+\epsilon/(3n)\right)}}
\end{align*}
and we solve
\begin{align*}
\left( \frac{\hat{K}+2}{\hat{K}+1} \right) \frac{\epsilon'}{3} + \frac{\epsilon'}{3} + \left( \frac{\hat{K}+2}{\hat{K}+1} \right) \frac{\epsilon'}{3} \leq \epsilon.
\end{align*}
for $\epsilon'$. Since 
\begin{align*}
\left( \frac{K+2}{K+1} \right) \frac{\epsilon'}{3} + \frac{\epsilon'}{3} + \left( \frac{K+2}{K+1} \right) \frac{\epsilon'}{3} \leq \left( \frac{\hat{K}+2}{\hat{K}+1} \right) \frac{\epsilon'}{3} + \frac{\epsilon'}{3} + \left( \frac{\hat{K}+2}{\hat{K}+1} \right) \frac{\epsilon'}{3}
\end{align*}
such a choice of $\epsilon'$ automatically satisfies
\begin{align*}
\left( \frac{K+2}{K+1} \right) \frac{\epsilon'}{3} + \frac{\epsilon'}{3} + \left( \frac{K+2}{K+1} \right) \frac{\epsilon'}{3} \leq \epsilon
\end{align*}
and this completes the proof.
\end{proof}

\begin{corollary}
{\bf P = PPAD}
\end{corollary}

\begin{proof}
 \citep[Theorem 15]{Avramopoulos2} shows that symmetric equilibrium approximation in fully polynomial time in symmetric games implies that {\bf P = PPAD}. Hence the corollary.
\end{proof}

\section*{Acknowledgments}

This paper has benefitted from my interaction with my YouTube account and I thank those that are responsible for the configuration of the content in that account.

\bibliographystyle{abbrvnat}
\bibliography{real}

\end{document}